\documentclass[final,1p,times,authoryear]{elsarticle}
\usepackage{amsmath}
\usepackage{amsthm}
\usepackage{amssymb}
\usepackage{amsfonts}

\newtheorem{theorem}{Theorem}
\newtheorem{lemma}[theorem]{Lemma}

\begin{document}

\begin{frontmatter}

\title{A proof of Hilbert's theorem on ternary quartic forms with the ladder technique}
\tnotetext[mytitlenote]{This work was partially supported by National Nature Science Foundation (61273007, 61673016), SWUN
Construction Projects for Graduate Degree Programs (2017XWD-S0805), Advance Research Program of
Electronic Science and Technology National Program (2017YYGZS16), Innovative Research Team of the
Education Department of Sichuan Province (15TD0050), Sichuan Youth Science and Technology Innovation
Research Team (2017TD0028).}

\author{Jia Xu}
\address{Computer Science and Technology Collage of Southwest University for Nationalities, Chengdu, China}
\ead{xufine@163.com}

\author{Yong Yao}
\address{Chengdu Institute of Computer Applications, Chinese Academy of Sciences, Chengdu, China}
\ead{yaoyong@casit.ac.cn}

\begin{abstract}
Hilbert's theorem states that every positive semi-definite real ternary quartic form can be written as a sum of squares of
quadratic forms. In this paper, we give a constructive proof with the rudiments of real analysis and the ladder technique. In order to ensure that sum of squares of quadratic forms can be practically constructed, the computation of real zero points is also discussed.
\end{abstract}

\begin{keyword}
Hilbert's theorem\sep positive semi-definite\sep ternary quartic forms\sep sum of squares\sep the ladder technique
\MSC[2010] 11E25 \sep 14P05
\end{keyword}
\end{frontmatter}

\section{Introduction}
In 1888, David Hilbert published a paper \citep{Hilbert} on the problem whether a positive semi-definite real polynomial is inevitably a sum of squares of other real polynomials. This work is influential and inspires a lot of great works in researchers even today. In the exceptional case, Hilbert proved that a positive semi-definite real ternary quartic form can be written as a sum of squares of
quadratic forms. Hilbert's proof is very brief but rather hard, because some complicated mathematical tools are used (see comments in \cite{Pfister} and \cite{Choi1}). Furthermore, Hilbert's method did not lend itself to a really practical construction. In
1977, Choi and Lam \citep{Choi1} showed a graceful elementary proof, in which only the rudiments of real analysis and the representation theorem of convex set are utilized. However,this method is not constructive either. Compared with \cite{Choi1}, though there are some similarities in the rudiments of real analysis, the proof of the present paper has the following two conspicuous differences. First, we take the ladder technique instead of the representation theorem of convex set. Furthermore, the method
of our proof is constructive whereas the one in \cite{Choi1} is not, since the representation theorem of convex set is an existence
theorem. Thus, for a positive semi-definite ternary quartic form, its explicit representation as a sum of squares can be constructed step by step according to the method addressed in this paper. From this perspective, we emphasize that this is not a completely new proof but, rather, a completely constructive account of the proof in \cite{Choi1}.

The rest of this paper is organized as follows. In section 2, four lemmas are presented. Each of them is like a step of the ladder, and the highest one (Lemma 4) achieves the goal. With the ladder technique, the Hilbert's theorem is proved easily. In section 3, the computation of real zero points is discussed, which provides the concrete details for the construction of sum of squares.

\section{Proof of Hilbert's theorem}

Firstly, the following are required. A form $f\in \mathbb{R}[x,y,z]$ is said to be positive semi-definite (psd) if
$$f(P)\geq 0,\ \forall P\in \mathbb{R}^3,$$ and ${\rm PSD}^4_3$ represents the set of all positive semi-definite ternary quartic forms. SOS
stands for the set of sum of squares of polynomials.

Given $P_1, P_2\in \mathbb{R}^3 \setminus \{ (0,0,0)\}$ satisfying $f(P_1)=f(P_2)=0$, we call $P_1$ and $P_2$ are same if there is a non-zero
number $\lambda$ such that $P_1=\lambda P_2$. That is, we discuss the zeros of $f$ in the real projective space $\mathbf{P}^2(\mathbb{R})$. The
set of  zeros in $\mathbf{P}^2(\mathbb{R})$ is denoted by $Z(f)$, and $|Z(f)|$ stands for the number of elements of $Z(f)$. The unit sphere is written as
$S^2=\{(x,y,z)|\ x^2+y^2+z^2=1\}$.

Next, we will present four lemmas.

\begin{lemma}
If the psd quartic form $f\in \mathbb{R}[x,y,z]$ has no zero in $\mathbf{P}^2(\mathbb{R})$ (i.e. $|Z(f)|=0$), then there is a quadratic form $g\in \mathbb{R}[x,y,z]$
such that $f-g^2$ is positive semi-definite and has at least one zero in $\mathbf{P}^2(\mathbb{R})$ (i.e. $|Z(f-g^2)|\geq 1$).
\end{lemma}

\begin{proof}
Since $f$ is continuous on the unit sphere $S^2$, $f$ inevitably attains a minimum $\lambda$. It is obvious that $\lambda>0$ and
$$f-\lambda(x^2+y^2+z^2)^2\geq 0.$$ Let $g=\sqrt{\lambda}(x^2+y^2+z^2)$, and then $f-g^2$ is positive semi-definite and $|Z(f-g^2)|\geq 1$.
\end{proof}

\begin{lemma}
If the psd quartic form $f\in \mathbb{R}[x,y,z]$ has only one zero in $\mathbf{P}^2(\mathbb{R})$ (i.e. $|Z(f)|=1$), then there are quadratic forms $g_1,g_2\in \mathbb{R}[x,y,z] $ and non-negative real numbers $a$ and $b$ with $a^2+b^2\neq 0$
such that $f-ag_1^2-bg_2^2$  is positive semi-definite and $|Z(f-ag_1^2-bg_2^2)|\geq 2$.
\end{lemma}

\begin{proof}
By coordinates transformation , we can assume that $f(1,0,0)=0$, and $f$ can be written as
\begin{equation}
f=x^2p(y,z)+2xq(y,z)+r(y,z),\label{eq:1}
\end{equation}
where $p\geq0,\ r\geq 0$, and $pr-q^2\geq 0$. There are three cases of $p$ to be discussed.

Case (i): $p=0$ (~0 form). It is clearly that~$q$ is necessarily a 0 form. Then
$$f=r(y,z).$$
Since~$f$ has only one zero $(1,0,0)$, it holds that $r(y,z)>0$ for all $(y,z)\in S^1=\{(y,z)\in \mathbb{R}^2|\ y^2+z^2=1\}.$
Furthermore,~$f$ is continuous on the unit circle~$S^1$, so that~$f$
has the minimum~$\lambda$ on $S^1$. It is obviously that~$\lambda>0$ and
$$f-\lambda(y^2+z^2)^2\geq 0.$$
Let~$g=\sqrt{\lambda}(y^2+z^2)$. Thus $f-g^2$ is positive semi-definite and $|Z(f-g^2)|\geq 2$.

Case (ii): The rank of $p$ is~$1$. By coordinate transformation, assume that $p=y^2$. The condition $pr-q^2\geq 0$ implies
that $q$ is divided by $y$ (the coefficient of $z^3$ in~$q$  is 0). That is, there is $q_1$ satisfying~$q=yq_1$.
 Substituting it into (\ref{eq:1}) one has
$$f=x^2y^2+2xyq_1+r=(xy+q_1)^2+r-q_1^2.$$
Notice that $r-q_1^2\ (\geq 0)$ is not a 0 form, otherwise~$f$ would have infinite zeros on $S^2$. There are two situations to be considered.

 (a) ~$r-q_1^2$ has a zero ~$(y_0,z_0)$ on ~$S^1$, then it has at least 2 zeros~$(1,0,0)$ and~$(0,y_0,z_0)$ on
~$S^2$. Hence $f-(xy+q_1)^2$ is positive semi-definite and ~$|Z(f-(xy+q_1)^2)|\geq 2$.

 (b) ~$r-q_1^2$ has no zero on~$S^1$, then ~$f-(xy+q_1)^2$ yields case (i), so that there is a quadratic form~$g$ such that~$f-(xy+q_1)^2-g^2$ is positive semi-definite and~$|f-(xy+q_1)^2-g^2|\geq 2$.

Case (iii): The rank of $p$ is $2$. By coordinate transformation, $p$ may take the form $p=y^2+z^2$. Consequently, we claim that $pr-q^2>0$
 on~$S^1$. Otherwise if $pr-q^2|_{(y_0,z_0)}=0$ for $(y_0,z_0)\neq (0,0)$, then $(-q(y_0,z_0)/p(y_0,z_0),y_0,z_0)$ is a zero of $f$, conflicting with the fact that $f$ only has the zero $(1,0,0)$. Thus~$pr-q^2$ attains the minimum $\lambda>0$ on $S^1$. As a result,
$$(pr-q^2)-\lambda(y^2+z^2)^3\geq 0,$$
with equality at~$(y_0,z_0)\in S^1$. Hence
$$f-\lambda(y^2+z^2)^2$$
 is still positive semi-definite by discriminant, and has at least two zeros: $(1,0,0)$ and $(\frac{-q(y_0,z_0)}{p(y_0,z_0)},y_0,z_0)$.

 It follows from case (i) to case (iii) that Lemma 2 holds.
\end{proof}

\begin{lemma}
If the psd quartic form $f\in \mathbb{R}[x,y,z]$ has only two zeros in $\mathbf{P}^2(\mathbb{R})$ (i.e. $|Z(f)|=2$), then there are  quadratic forms $g_1,g_2,g_3\in \mathbb{R}[x,y,z] $ and non-negative real numbers $a,b,c$ with $a^2+b^2+c^2\neq 0$
such that $f-ag_1^2-bg_2^2-cg_3^2$   is positive semi-definite and $|Z(f-ag_1^2-bg_2^2-cg_3^2)|\geq 3$.
\end{lemma}

\begin{proof}
Change coordinates so that~$f(1,0,0)=f(0,1,0)=0$. Write
\begin{equation}
f=x^2p(y,z)+2xzq(y,z)+z^2r(y,z),\label{eq:2}
\end{equation}
where $p,\ q$ and $r$ are quadratic forms with ~$p\geq0,\ r\geq 0,$ and $pr-q^2\geq 0$. Next the proof will be split into  two cases.

Case (i): If at least one of $p$ and $r$ has a zero on $S^1$, then let $(y_0,z_0)$ be the zero of $r$ without loss of generality. Substituting it into (\ref{eq:2}), one yields
$$f(x,y_0,z_0)=x^2p(y_0,z_0)+2xz_0q(y_0,z_0).$$
Then~$(0, y_0,z_0)$  is also a zero of~$f$, hence $y_0=1,z_0=0$. Substituting them into quadratic form $r(y,z)$, it is reduced to the form
$$r(y,z)=tz^2,\ t> 0.$$ From~$pr-q^2=tpz^2-q^2\geq 0$, $q$ must be divided by $z$. Thus~$q=\sqrt{t}zq_1$,
where $q_1$ is a linear form. Therefore,
$$pr-q^2=tz^2(p-q_1^2)\geq 0\Longrightarrow p-q_1^2\geq 0,$$ and
$$f=x^2p+2xz^2\sqrt{t} q_1+z^2(tz^2)=(\sqrt{t}z^2+xq_1)^2+x^2(p-q_1^2).$$
Then~$f-(\sqrt{t}z^2+xq_1)^2=x^2(p-q_1^2)$ is positive semi-definite and has infinite zeros~$(0,y,z)$. Hence
$$|Z(f-(\sqrt{t}z^2+xq_1)^2)|>3.$$

Case (ii): If neither $p$ nor $r$ is supposed to have a zero on~$S^1$ (both $p$ and $r$ are strictly positive on $S^1$), then discuss two possibilities of the discriminant~$pr-q^2$.

(a) The discriminant~$pr-q^2$ has a zero~$(y_0,z_0)\in S^1$, and let $\lambda=-\frac{q(y_0,z_0)}{p(y_0,z_0)}$. Then one has
$$f_1(x,y,z)=f(x+\lambda z, y,z)=x^2p+2xz(q+\lambda p)+z^2(r+2\lambda q+\lambda^2p).$$
Since~$(y_0,z_0)$ is a zero of~$r+2\lambda q+\lambda^2p$,
$f_1$ is reduced to the case (i).

(b) The discriminant~$pr-q^2$ has no zero on~$S^1$. Since~$p$ and $r$ are strictly positive on~$S^1$, the function
$$\frac{pr-q^2}{p(y^2+z^2)}$$
is also strictly positive on~$S^1$, and its minimum~$\lambda>0$. Thus
$pr-q^2\geq \lambda p(y^2+z^2),$
and~$f-\lambda z^2(y^2+z^2)$  is positive semi-definite. If~$r-\lambda(y^2+z^2)$ has zeros, then~$f-\lambda z^2(y^2+z^2)$ is reduced to the case (i). In contrast, if~$r-\lambda(y^2+z^2)$ has no zero on~$S^1$, then $f-\lambda z^2(y^2+z^2)$ is reduced to the case (ii) (a).

Because we have considered all two cases, we can conclude that Lemma 3 holds.
\end{proof}

\begin{lemma}
If the number of zeros of the psd quartic form $f\in \mathbb{R}[x,y,z]$ in $\mathbf{P}^2(\mathbb{R})$ is more than 3 (i.e. $|Z(f)|\geq 3$), then $f$
is a sum of squares of quadratic forms.
\end{lemma}

\begin{proof}
The proof will be split into  two cases.

Case(i): If there are three zeros not on the same line, then by arranging coordinates, we may assume that
$$f(1,0,0)=f(0,1,0)=f(0,0,1)=0.$$
This implies that the degree of each variable is less than or equal to $2$. Write
$$f=ax^2y^2+by^2z^2+cz^2x^2+a_1x^2yz+b_1y^2xz+c_1z^2xy,\quad a,b,c,a_1,b_1,c_1\in \mathbb{R}.$$
Do substitution $xy=X,yz=Y,zx=Z$, and $f$ can be transformed into a positive semi-definite form in $X,Y,Z$. It is obvious that $f$ is a sum of squares of quadratic forms.

Case (ii): If all of the zeros are on the same line, then it can be split into two subcases.

(a) $f$ has infinite zeros. Then $f$ has linear factors by B$\acute{e}$zout's theorem. Since $f$ is positive semi-definite, the degree of the linear factors is even. Hence they are $$l^2h_2,\ l^4,$$ where $l$ is linear form and $h_2$ is positive semi-definite quadratic form (maybe degenerate). It is obvious that all of them are sum of squares of quadratic forms.

(b) $f$ has finite zeros. we claim that this case will never happen. The reason is as follows.

Arrange coordinates so that $f(1,0,0)=f(0,1,0)=f(x_0,y_0,0)=0$ with $x_0y_0 \neq 0$. Write
\begin{equation}
f=x^2p(y,z)+2xzq(y,z)+z^2r(y,z),\label{eq:3}
\end{equation}
where $p,q$ and $r$ are quadratic forms with~$p\geq0, r\geq 0,  pr-q^2\geq 0$. Therefore,
$$f(x_0,y_0,0)=x_0^2p(y_0,0)=0.$$
Furthermore, $x_0\neq 0$ yields $p(y_0,0)=0$, where $p$ is  a quadratic form only in $y$ and $z$. Let
$$p=ay^2+2byz+cz^2, a,b,c\in \mathbb{R}.$$
Consequently $p(y_0,0)=ay^2_0=0$, and $a=0$ by $y_0\neq 0$. Thus $z$ divides~$f$ by (\ref{eq:3}), and all of the points on the line $z=0$ are zeros of $f$. This contradicts the premise that $f$ has finite zeros on the line $z=0$.
\end{proof}

With the above 4 lemmas,  Hilbert's theorem will be proved with a chart.

\newtheorem{thm}{Theorem}
\begin{thm}[Hilbert]
A positive semi-definite ternary quartic form over the reals can be written as a sum of squares of quadratic forms.
\end{thm}

\begin{proof}
The following chart will present the process of proof.
$$
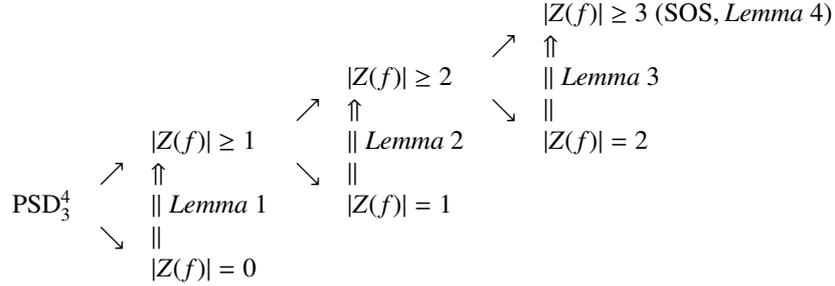
\begin{array}{lllllll}
& & & & & &|Z(f)|\geq 3\ ({\rm SOS}, Lemma \ 4)\\
& & & & & \nearrow & \Uparrow \\
& & & & |Z(f)|\geq 2 &  & \| \ Lemma\ 3 \\
& & &\nearrow & \Uparrow  & \searrow & \|  \\
& &|Z(f)|\geq 1 & &\| \ Lemma\ 2 & &|Z(f)|=2 \\
& \nearrow &\Uparrow &\searrow &\|  & & \\
{\rm PSD}^4_3 &  &\| \ Lemma\ 1  &  &|Z(f)|=1  & & \\
& \searrow &\| & &  & & \\
&  &|Z(f)|=0 & &  & &
\end{array}
$$
The set ${\rm PSD}^4_3$ is divided into two disjoint subsets, that is, the subset with $|Z(f)|=0$ and the other one with $|Z(f)|\geq 1$. According to Lemma 1, the former can be transformed into the latter by minus a square of a quadratic form. Furthermore, the subset with $|Z(f)|\geq 1$ can be dealt with in the same way according to Lemma 2 and so forth. Finally the theorem holds according to Lemma 4.
\end{proof}

\newdefinition{rmk1}{Remark}
\begin{rmk1}
Theorem 1 is the weak edition of Hilbert's theorem. The strong edition is that the positive semi-definite ternary quartic forms are sum of no more than three squares of quadratic forms. The proof of the strong edition can see \cite{Pfister}, \cite{Powers} and \cite{SC2}.
\end{rmk1}

\newdefinition{rmk2}[rmk1]{Remark}
\begin{rmk2}
The real zeros are assumed to be known in the above proof. To achieve the goal of really practical construction , we still need some method in solving the real zeros of ternary quartic forms. This problem will be conquered in Section 3, and we will really obtain the practically constructive proof of Hilbert's theorem.
\end{rmk2}

Next we will present an example.
\newdefinition{ex1}{Example}
\begin{ex1}\citep{Cirtoaje}
Compute the sum of squares of the following ternary quartic form.
 \begin{eqnarray*}
 &f=&4(x^4+y^4+z^4)+21(xy+yz+zx)^2-10(x^2+y^2+z^2)(xy+yz+zx)\\
 & &-37xyz(x+y+z).
 \end{eqnarray*}

\emph{Solution:}
 Since $f$ has 4 zeros in the projective space. There are
 \begin{eqnarray*}
 (1,1,1),\quad (3,2,2),\quad (2,3,2),\quad (2, 2, 3).
 \end{eqnarray*}
 Select three of them and construct the following matrix (each column  is a zero of $f$)
 $$
A=\left [
\begin{matrix}
1&\  3& \ 2 \\
 1& \ 2& \ 3 \\
 1& \  2& \ 2
\end{matrix}
\right ].
$$
Do linear transformation
$$
\left[
\begin{matrix}
x\\
y\\
z
\end{matrix}
\right ]
=A
\left[
\begin{matrix}
\bar{x}\\
\bar{y}\\
\bar{z}
\end{matrix}
\right ].
$$
By computing, we have
\begin{eqnarray}
f(A[\bar{x},\bar{y},\bar{z}]^T)&=&\bar{x}^2\bar{y}^2+49\bar{y}^2\bar{z}^2+\bar{z}^2\bar{x}^2-\bar{x}^2\bar{y}\bar{z}+7\bar{y}^2\bar{z}\bar{x}
+7\bar{z}^2\bar{y}\bar{x} \nonumber \\
&=&\frac{1}{2}\left( (\bar{x}\bar{y}-\bar{z}\bar{x})^2+(\bar{x}\bar{y}+7\bar{y}\bar{z})^2+(\bar{z}\bar{x}+7\bar{y}\bar{z})^2\right).\label{eq:4}
\end{eqnarray}
Compute the inverse matrix of $A$,
$$
A^{-1}=\left [
\begin{matrix}
 -2& \ -2&\ \  5 \\
 \ 1& \ \ 0&\ -1 \\
  \ 0& \ \ 1&\ -1
\end{matrix}
\right ].
$$
Let
$$
\left[
\begin{matrix}
\bar{x}\\
\bar{y}\\
\bar{z}
\end{matrix}
\right ]
=A^{-1}
\left[
\begin{matrix}
x\\
y\\
z
\end{matrix}
\right ]
=
\left[
\begin{matrix}
-2x-2y+5z\\
x-z\\
x-y
\end{matrix}
\right ],
$$
and substitute it into (\ref{eq:4}), and then the sum of squares of $f$ is as follows.
$$f=\frac{1}{2}\sum (-2x^2+5xz+2y^2-5yz)^2.$$
\end{ex1}

\section{Computation of zeros}
In this section, we will discuss the problem of computing zeros. Lemma 4 presents a key of constructing sum of squares, that is, finding out at least three real zeros of a positive semi-definite ternary quartic form. Next we will prove a proposition (Lemma \ref{lem5}) that can solve the problem of computing  the real zeros of positive semi-definite ternary forms. The fundamental idea comes from \cite{Yang,Xia}.

\begin{lemma}\label{lem5}
Given positive semi-definite form $f\in \mathbb{R}[x,y,z]$, let $f'_x$ be the derivative of $f$ with respect to $x$, and ${\rm res}(f, f'_x, x)$ be the resultant of~$f$ and $f'_x$ with respect to $x$. Then the equation $\{f=0\}$ and the equations $\{f=0,\ {\rm res}(f,f'_x, x)=0\}$ are equivalent in the projective space $\mathbf{P}^2(\mathbb{R})$.
\end{lemma}

\begin{proof}
It is obvious that $Z(f)\supseteq Z(f,\ {\rm res}(f,f'_x, x))$. Next we will prove $Z(f)\subseteq Z(f,\ {\rm res}(f,f'_x, x)).$

It is easy to get that the real zero $(x_0,y_0,z_0)\in S^2$ of positive semi-definite form $f$ need to satisfy the following equations for solving stationary points.
\begin{equation}
\left \{
\begin{array}{l}
f'_x=0,\\
f'_y=0,\\
f'_z=0.
\end{array}
\right .
\end{equation}
This is because of the fact that if at least one of the following values
$$f'_x|_{(x_0,y_0,z_0)},\ f'_y|_{(x_0,y_0,z_0)},\ f'_z|_{(x_0,y_0,z_0)}$$
is not zero, then $f$ can be approximated by linear functions on a neighborhood of the point $(x_0,y_0,z_0)$. So
$f$ is not nonnegative on a neighborhood of the point $(x_0,y_0,z_0)$. This contradicts the premise that $f$ is positive semi-definite. Thus $x_0$ satisfies the following equations.
\begin{equation}
\left \{
\begin{array}{l}
f(x,y_0,z_0)=0,\\
f'_x(x,y_0,z_0)=0.
\end{array}
\right .
\end{equation}
Then ${\rm res}(f,f'_x,x)|_{(y_0,z_0)}=0$ according to the fundamental property of resultant. That is, $Z(f)\subseteq Z(f, {\rm res}(f,f'_x, x)).$
\end{proof}

\newdefinition{rmk3}{Remark}
\begin{rmk3}
In Lemma \ref{lem5} $f$ cannot include square factors, otherwise ${\rm res}(f,f'_x,x)$ would identically equal to 0. Thus the equations $\{f=0, {\rm res}(f,f'_x,x)=0\}$ in $\mathbf{P}^2(\mathbb{R})$ is zero dimension, and we can solve the above equations with various efficient methods such as method of resultant \citep{Kapur,Canny}, rational single variable present method \citep{Rouillier}, etc.
\end{rmk3}

\newdefinition{rmk4}[rmk3]{Remark}
\begin{rmk4}
The general version of Lemma 5 (successive resultant method \citep{Yang,Xia}) also ensure the minimum of Lemma 1, Lemma 2 and Lemma 3 can be computed accurately, and then ensure the above proof of Hilbert's theorem is completely constructive. However, the following example (example 3) shows that the practical computation may be very complex and finally comes to nothing because of high time complexity.
\end{rmk4}

Next we will present two computing examples.

\newdefinition{ex2}[ex1]{Example}
\begin{ex2}(\textbf{Vasile Cirtoaje})\citep{Cirtoaje}
Compute real zeros of $$ f=(x^2+y^2+z^2)^2-3(x^3y+y^3z+z^3x).$$

\emph{Solution:} Firstly compute
$$f'_x=4(x^2+y^2+z^2)x-9x^2y-3z^3.$$
Then compute the resultant ${\rm res}(f,f'_x,x)$,
$${\rm res}(f,f'_x,x)=9(13y^4-18y^3z-y^2z^2-6yz^3+13z^4)(-z+y)^2(y^3-5y^2z+6yz^2-z^3)^2.$$
By~Gram matrix method \citep{Choi2}, one yields
\begin{eqnarray*}
h &=&(13y^4-18y^3z-y^2z^2-6yz^3+13z^4)\\
 &=&13(y^2-\frac{9}{13}yz-\frac{7}{13}z^2)^2+286(\frac{2}{13}yz-\frac{51}{286}z^2)^2+\frac{3}{22}z^4.
\end{eqnarray*}
Thus $h$ is strictly positive on $S^1$. Hence
$${\rm res}(f,f'_x,x)=0 \Longleftrightarrow (-z+y)(y^3-5y^2z+6yz^2-z^3)=0.$$
By using real root isolation algorithm \citep{Xia}, we know that the equation ${\rm res}(f,f'_x,x)$ $=0$ has four real zeros in the projective space $\mathbf{P}^2(\mathbb{R})$:
$$(1,1),\ (\alpha_1, 1),\ (\alpha_2, 1),\ (\alpha_3, 1),$$
where $\alpha_1<\alpha_2<\alpha_3$ are three positive real zeros of $t^3-5t^2+6t-1=0$. Substituting the above four points into the polynomial $f$, and we have four real zeros of $f=0$ in the real projective space $\mathbf{P}^2(\mathbb{R})$. There are
\begin{equation}
(1,1,1),\ (\frac{1}{\alpha_2},\alpha_1,1),\ (\frac{1}{\alpha_3},\alpha_2,1),
\  (\frac{1}{\alpha_1},\alpha_3,1).
\end{equation}
Consequently $f$ can be written as a sum of squares according to Lemma 4.
\end{ex2}

\newdefinition{ex3}[ex1]{Example}
\begin{ex3}(\textbf{C. Scheiderer})\citep{Scheiderer}
$g=x^4+y^4+z^4+xy^3+xz^3+yz^3-3x^2yz-4xy^2z+2x^2y^2.$

In \cite{Scheiderer} this form is proven to have the following characters: $g$ cannot be written as sum of squares of quadratic forms with rational coefficients. That is to say the answer of Sturmfels's question \citep{Scheiderer} is negative. It means that it is very difficult to write $g$ as sum of squares for the tools that depend on numerical computation such as SOSTOOLS \citep{Parrilo}.

Next we compute the zeros of $g$ in the real projective space.

Firstly we compute
$$g'_x=4x^3+4xy^2-6xyz+y^3-4zy^2+z^3.$$
Then compute the resultant
\begin{align*}
\begin{array}{lll}
{\rm res}(g,g'_x,x)&=&229y^{12}-1904y^{11}z+5896y^{10}z^2+1376y^9z^3-12176y^8z^4\\
& &+6432y^7z^5+8630y^6z^6-9472y^5z^7+952y^4z^8+3232y^3z^9\\
& &-96y^2z^{10}+336yz^{11}+229z^{12}.
\end{array}
\end{align*}
By using real root isolation algorithm, we know that ${\rm res}(g,g'_x,x)$ only has one trivial zero $(0,0)$, but~$(1,0,0)$
is not the zero of~$g$. Hence~$Z(g)=\emptyset$ (empty set), and $g$ is strictly positive on~$S^2$.

By Lemma 1, one yields that there is a positive constant $t$ such that
$$g_t=g-t(x^2+y^2+z^2)^2$$
is still positive semi-definite and has at least a zero in the real projective space $\mathbf{P}^2(\mathbb{R})$. Furthermore, by using successive resultant algorithm \citep{Yang,Xia}, one can compute $t$ accurately. $t$, in the interval $[51/512, 103/1024]$ and the approximation being $0.10009018$, is the real zero of the following equation with degree 12.
\begin{align*}
\begin{array}{ll}
&1540909743009169408\ x^{12}-13437733654176464896\ x^{11}\\
&+51805978528683065344\ x^{10}-116396366581901484032\ x^9\\
&+168975565335348900096\ x^8-165910705322168135008\ x^7\\
&+111957978056509355125\ x^6-51652982930080321180\ x^5\\
&+15876922302830413280\ x^4-3088008227838928440\ x^3\\
&+347409936566531728\ x^2-19347901948050048\ x\\
&+380514157362176=0.
\end{array}
\end{align*}
It is very difficult to compute the real zeros of $g_t$ for high time complexity.

In \cite{Scheiderer} $g$ can be written as
$$g=\frac{1}{4}\left ((2x^2+\beta y^2-yz+(2+\frac{1}{\beta})z^2 )^2
-\beta (2xy-\frac{y^2}{\beta}+\frac{2xz}{\beta}+\beta yz-z^2 )^2\right ),$$
where $\beta$ is a negative zero of equation~$t^4-t+1=0$.
\end{ex3}

\section{Conclusion}
In this paper, we give a proof of Hilbert's theorem by four lemmas and the ladder technique. According to the proof, catching at least three real zeros of a positive semi-definite ternary quartic form is necessary for constructing its sum of squares. Consequently, we present the method of locating the zeros based on the property that zeros of a positive semi-definite form satisfy the equations for stationary points. However the representation of sum of squares using this method may be complex and how to build a clear representation is a new question.

\bibliographystyle{elsarticle-harv}

\end{document}